%% file: main.tex
\documentclass{article}
\usepackage{microtype}
\usepackage{amsthm}
\usepackage{amsmath}
\usepackage{amsfonts}
\usepackage{float}
\usepackage{mathrsfs}

    \makeatletter
    \let\@fnsymbol\@arabic
    \makeatother

\title{The Stochastic Score Classification Problem}

\author{Dimitrios Gkenosis\thanks{Department of Informatics and Telecommunications, University of Athens, Athens, Greece \newline \texttt{gkenosis.dimitrios@math.uoa.gr}} \and Nathaniel Grammel\thanks{Department of Computer Science, University of Maryland, College Park, Maryland, US \newline \texttt{ngrammel@cs.umd.edu }} \and Lisa Hellerstein\thanks{NYU Tandon School of Engineering, Brooklyn, NY, USA \newline \texttt{lisa.hellerstein@nyu.edu}} \and Devorah Kletenik\thanks{Department of Computer and Information Science, Brooklyn College, CUNY, Brooklyn, New York, US \newline \texttt{kletenik@sci.brooklyn.cuny.edu}}}
\usepackage{hyperref}
\usepackage{graphicx}

\usepackage{algorithm}
\usepackage{algorithmic}

\usepackage{tabularx}
\newcolumntype{Y}{>{\centering\arraybackslash}X}

\DeclareMathOperator*{\argmin}{arg\,min} 

\theoremstyle{plain}
\newtheorem{claim}{Claim} 
\newtheorem{theorem}{Theorem}
\newtheorem{lemma}{Lemma}
\newtheorem{corollary}{Corollary}

\newtheorem{remark}{Remark}

\newcommand{\bit}{\{ 0, 1 \}}
\newcommand{\fullassign}{\bit^n}
\newcommand{\partassign}{\{0,1,*\}^n}
\newcommand{\Ints}{\mathbb{Z}}

\newcommand{\nonnegInts}{\Ints_{\geq 0}}
\newcommand{\boolfunc}[1][f]{\begingroup #1 \colon \bit^{n}\to \bit \endgroup}
\newcommand{\classfunc}[1][f]{\begingroup #1 \colon \bit^{n}\to \{1,\ldots,B\} \endgroup}
\newcommand{\utilfunc}[1][g]{\begingroup #1 \colon \partassign\to \nonnegInts \endgroup}
\newcommand{\OPT}{\mathsf{OPT}}
\newcommand{\ALG}{\mathsf{ALG}}

\renewcommand{\arraystretch}{1.75}

\begin{document}

\maketitle

\begin{abstract}
Consider the following Stochastic Score Classification Problem.  A doctor is assessing a patient's risk of developing a certain disease, and can perform $n$ tests on the patient.  Each test has a binary outcome, positive or negative.  A positive test result is an indication of risk, and a patient's score is the total number of positive test results.  
The doctor needs to classify the patient into one of $B$ risk classes, depending on the score (e.g., LOW, MEDIUM, and HIGH risk).  Each of these classes corresponds to a contiguous range of scores.  
Test $i$ has probability $p_i$ of being positive, and it costs $c_i$ to perform the test.  To reduce costs, instead of performing all tests, the doctor will perform them sequentially and stop testing when it is possible to determine the risk category for the patient.  The problem is to determine the order in which the doctor should perform the tests, so as to minimize the expected testing cost.  We provide approximation algorithms for adaptive and non-adaptive versions of this problem, 
and pose a number of open questions.  
\end{abstract}

\section{Introduction}
We consider the following Stochastic Score Classification (SSClass) problem.
A doctor can perform $n$ tests on a patient, each of which has a positive or negative outcome.   Test $i$ has known probability $p_i$ of having a positive outcome, and costs $c_i$ to perform.  A positive test is indicative of the disease.  The professor needs to assign the patient to a risk class (e.g., LOW, MEDIUM, HIGH) based on how many of the $n$ tests are positive. 
Each class corresponds to a contiguous range of scores.

To reduce costs, instead of performing all tests and computing an exact score, the doctor will perform them one by one, stopping when the class becomes a foregone conclusion.
For example, suppose there are 10 tests and the MEDIUM class corresponds to a score between 4 and 7 inclusive.  If the doctor performed 8 tests, of which 5 were positive, the doctor would not perform the remaining 2 tests, because the final score will be between 5 and 7, meaning that the risk class will be MEDIUM regardless of the outcome of the 2 remaining tests.
The problem is to compute the optimal (adaptive or non-adaptive) order in which to perform the tests, so as to minimize expected testing cost. 

Formally, the Stochastic Score Classification problem is as follows.  Given $B+1$ integers $0=\alpha_1 < \alpha_2<  \ldots < \alpha_{B} <\alpha_{B+1}=n+1$, 
let {\em class $j$} correspond to the {\em scoring interval} $\{\alpha_j,\alpha_j+1, \ldots, \alpha_{j+1}-1 \}$.
The $\alpha_j$ define an associated  
pseudo-Boolean \emph{score classification} function  $f:\{0,1\}^n \rightarrow \{1,\ldots,B\}$, such that $f(X_1, \ldots, X_n)$
is the class whose scoring interval contains the \emph{score} $r(X) = \sum_i X_i$.
Note that $B$ is the number of classes.
Each input variable $X_i$ is independently $1$ with given probability $p_i$, where $0 < p_i < 1$, and is $0$ otherwise.  The value of $X_i$ can only be determined by asking a query (or performing a test), which incurs a given non-zero, real-valued cost $c_i$.

An \emph{evaluation strategy} for $f$ is a sequential adaptive or non-adaptive ordering in which to ask the $n$ possible queries.  
Each query can only be asked once.  Querying must continue until the value of $f$ can be determined, i.e., until the value of $f$ would be the same, no matter how the remainder of the $n$ queries were answered.  
The goal is to design an evaluation strategy for $f$ with minimum expected total query cost. 

We consider both adaptive and non-adaptive versions of the problem.  In the adaptive version, we seek an adaptive strategy, where the choice of the next query can depend on the outcomes of previous queries.  An adaptive strategy corresponds to a decision tree, although we do not require the tree to be output explicitly (it may have exponential size).  In the non-adaptive version, we seek a non-adaptive strategy, which
is a permutation of the queries.   With a non-adaptive strategy, querying proceeds in the order specified by the permutation until the value of $f$ can be determined from the queries performed so far.  

We also consider a weighted variant of the problem, where query $i$ has given integer weight $a_i$, the score is $\sum_i a_i X_i$,
and $\alpha_1 < \alpha_2<  \ldots < \alpha_{B} <\alpha_{B+1}$ where $\alpha_1$ equals the minimum possible value of the score  $\sum_i a_i X_i$, and $\alpha_{B+1}-1$ equals the maximum possible score.  

While we have described the problem above in the context of assessing disease risk, such classification is also used in other contexts, such as assigning letter grades to students, giving a quality rating to a product, and deciding 
 whether or not a person charged with a crime should be released on bail.  In Machine Learning,  the focus is on learning the score classification function~\cite{ustun2016supersparse,tran2016preterm,jung2017simple,zeng2017interpretable,ustun2017optimized}. In contrast, here our focus is on reducing the cost of evaluating the classification function.  

Restricted versions of the weighted and unweighted SSClass problem have been studied previously.
In the algorithms literature, Deshpande et al.\ presented two approximation algorithms solving the {\em Stochastic Boolean Function Evaluation} (SBFE) problem for linear threshold functions~\cite{DeshpandeGoalValue}. 
The general SBFE problem is similar to the adaptive SSClass problem, but instead of evaluating a given score classification function $f$ defined by inputs $\alpha_j$, you need to evaluate a given Boolean function $f$.  When $f$ is a linear threshold function, the problem is equivalent to the weighted adaptive SSClass problem.  
One of the two algorithms of Deshpande et al.\ achieves an $O(\log W)$-approximation factor for this problem using the submodular goal value approach; it involves construction of a goal utility function and application of the Adaptive Greedy algorithm of Golovin and Krause to that function~\cite{golovinKrause}.
Here $W$ is the sum of the magnitudes of the weights $a_i$.
The other algorithm achieves a 3-approximation by applying a dual greedy algorithm to the same
goal utility function.  

A $k$-of-$n$ function is a Boolean function $f$ such that $f(x) = 1$ iff $x_1 + \ldots + x_n \geq k$. 
The SBFE problem for evaluating $k$-of-$n$ functions is equivalent to the unweighted adaptive SSClass problem, with only two classes ($B=2$).  
It has been studied previously in the VLSI testing literature. There is an elegant algorithm for the problem that computes an optimal strategy~\cite{Salloum79,BenDov81,SalloumBreuer84,Chang}.  

The unweighted adaptive SSClass problem for arbitrary numbers of classes was studied in the information theory literature~\cite{Dasetal12,Acharyaetal11,kowshikkumar13}, but only for unit costs.  The main novel contribution there was to establish an equivalence between verification and evaluation, which we discuss below.   


\section{Results and open questions}
We give approximation results for adaptive and non-adaptive versions of the SSClass problem.  We describe most of our results here, but leave description of some others and some of the proofs to the appendix.
A table with all our bounds can be found 
in the next section.

We begin by using the submodular goal value approach of Deshpande et al.\ to
obtain an $O(\log W)$ approximation algorithm for the weighted adaptive SSClass problem.  
This immediately gives an $O(\log n)$  approximation for the unweighted adaptive problem.  We also present a simple alternative algorithm achieving a $B-1$ approximation for the unweighted adaptive problem, and a
$3(B-1)$-approximation algorithm for the weighted adaptive problem again using an algorithm of
Deshpande et al.

We then present our two main results, which are both for the case of unit costs. The first is a 4-approximation algorithm for the adaptive and non-adaptive versions of the unweighted SSClass problem. The second is a $\varphi$-approximation for a special case of the non-adaptive unweighted version, where the problem is to evaluate what we call the Unanimous Vote Function.  Here $\varphi = \frac{1+\sqrt{5}}{2} \approx 1.618$ is the golden ratio.  The Unanimous Vote Function outputs
 POSITIVE if $X_1 = \ldots = X_n = 1$, NEGATIVE if $X_1 = \ldots = X_n = 0$, and UNCERTAIN otherwise.  Equivalently, it is a score classification function with $B=3$ and
scoring intervals $\{0\}, \{1, \ldots, n-1\}$ and $\{n\}$.  
The proofs of our two main results imply upper bounds of 4 and $\varphi$ for the adaptivity gaps of the corresponding problems.   

We use both existing techniques and new ideas in our algorithms.  We use the submodular goal value approach of Deshpande et al.\ to get our $O(\log W)$ bound for the weighted adaptive SSClass problem. This approach cannot yield a bound better than $O(\log n)$ for SSClass problems, since they involve evaluating a function of $n$ relevant Boolean variables~\cite{bachetalgoalvalue}.  

For our other bounds, we exploit the exact algorithm for $k$-of-$n$ evaluation, and the ideas used in its analysis.  
To obtain non-adaptive algorithms for the unit-cost case, we perform a round robin between 2 subroutines, one performing queries in increasing order of $c_i/p_i$, while the second performs them in increasing order of $c_i/(1-p_i)$. For arbitrary costs, instead of standard round robin, we use the modified round robin approach of Allen et al~\cite{allen2017evaluation}.
As has been repeatedly shown, the $c_i/p_i$ ordering and the $c_i/(1-p_i)$ ordering are optimal for evaluation of the Boolean OR (1-of-$n$) and AND ($n$-of-$n$) functions respectively (cf.~\cite{unluyurtReview}).
Intuitively, the first ordering (for OR) favors queries with low cost and high probability of producing the 
value 1, while the second (for AND) favors queries with low cost and high probability of producing the value 0. The proof of optimality follows from the fact that given any ordering, swapping two adjacent queries that do not follow the designated increasing order will decrease expected evaluation cost.

While the algorithm for our first main result is simple, the proof of its 4-approximation bound is not.  It uses ideas from the existing analysis of the $k$-of-$n$ algorithm, which is an easier problem because $B=2$.  To obtain our 4-approximation result we perform a new, careful analysis.  Unlike the analysis of the $k$-of-$n$ algorithm, this analysis only works for unit costs.  

To develop our $\varphi$-approximation for the unanimous vote function, we first note that for such a function, if you perform the first query and observe its outcome, the optimal ordering of the remaining queries can be determined by evaluating a Boolean OR function, or the complement of an AND function.   We then address the
problem of determining an approximately optimal permutation, given the first query.  A standard round robin between the $c_i/p_i = 1/p_i$ ordering, and the $1/(1-p_i)$ ordering, 
yields a factor of 2 approximation.  To obtain the $\varphi$ factor, we stop the round robin at a carefully chosen point and \emph{commit} to one of the two subroutines, abandoning the other.  Our full algorithm for the unanimous vote function works by trying all $n$ possible first queries.  For each, we generate the approximately optimal permutation, and algebraically compute its expected cost.  Finally, out of these $n$ permutations, we choose the one with lowest expected cost.  

We note that although our algorithms are designed to minimize expected cost for independent queries, the goal value function used to achieve the $O(\log W)$ approximation result can also be used to achieve a worst-case bound, and a related bound in the Scenario model~\cite{golovinKrause,grammel2016scenario,kambadur2017adaptive}.

A recurring theme in work on
SSClass problems has been the relationship between these evaluation problems and their
associated verification problems.
In the verification problem, you are given the output class (i.e., the value of the score classification function) before querying, and
just need to perform enough tests to certify (verify) that the given output class is correct.
Thus optimal expected verification cost lower bounds optimal expected evaluation cost.
Surprisingly, the result of 
Das et al.~\cite{Dasetal12} showed that for the adaptive SSClass problem in the unit-cost case,
optimal expected verification cost equals optimal expected evaluation cost.
Prior work already implied this was true for evaluating $k$-of-$n$ functions, even
for arbitrary costs (cf.~\cite{unluyurtBorosDoubleRegular}).
We give a counterexample in the full paper~\cite{} showing that this relationship does not hold for the adaptive SSClass problem with arbitrary costs.
Thus algorithmic approaches based on optimal verification strategies may not be effective for these problems.

There remain many intriguing open questions related to
SSClass problems.
The first, and most fundamental, is whether the (adaptive or non-adaptive) SSClass
problem is NP-hard.
This is open even in the unit-cost case.
It is unclear whether this problem will be easy to resolve.
It is easy to show that the weighted variants
are NP-hard: this follows from the NP-hardness of the SBFE problem for
linear threshold functions, which is proved by a simple reduction from knapsack~\cite{DeshpandeGoalValue}.
However, the approach used in that proof is to show that the deterministic version of
the problem (where query answers are known a-priori) is NP-hard, which is not the case in the SSClass problem.
Further, NP-hardness of evaluation problems is not always easy to determine.
The question of whether the SBFE problem for read-once formulas is NP-hard has been
open since the 1970's (cf.~\cite{Greiner06}).

Another main open question is whether there is a constant-factor approximation algorithm
for the weighted SSClass problem.  Our bounds depend on $n$ or $B$.
Other open questions concern lower bounds on approximation factors, and bounds on adaptivity gaps.

\section{Table of Results}
\label{append:table}

\begin{table}[H]
\caption{Results for the adaptive SSClass Problem}
\label{table:adaptive}
\begin{tabularx}{1.25\linewidth}{ X | X| X}
& unit costs & arbitrary costs \\
\hline
weighted & $O(\log W)$-approx [Sec. ~\ref{sec:general}];  \newline $3(B-1)$ [Sec. ~\ref{sec:general}]& 
$O(\log W)$-approx [Sec. ~\ref{sec:general}];  \newline $3(B-1)$ [Sec. ~\ref{sec:general}]\\
\hline
unweighted & 4-approx [Sec. ~\ref{sec:rrnonadaptive},~\ref{sec:4approx}] & $O(\log n)$-approx; 
\newline $(B-1)$-approx [Sec. ~\ref{sec:general},~\ref{sec:B-1approx}]\\ 
\hline
$k$-of-$n$ function & exact algorithm [known] & exact algorithm [known] \\
\hline
unanimous vote \newline function & exact algorithm [Sec.~\ref{sec:unanimousadaptive}] & exact algorithm [Sec.~\ref{sec:unanimousadaptive}]
\end{tabularx}
\end{table}

\begin{table}[ht]
\caption{Results for the non-adaptive SSClass problem}
\label{table:non-adaptive}
\begin{tabularx}{\textwidth}{ c | c | c}
& unit costs & arbitrary costs \\
\hline
weighted & open & open  \\
\hline
unweighted & 4-approx 
[Sec.~\ref{sec:rrnonadaptive},~\ref{sec:4approx}] & $2(B-1)$-approx [Sec.~\ref{sec:rrnonadaptive},~\ref{sec:nonadaptive2B-1}] \\
\hline
$k$-of-$n$ function & 2-approx [Sec.~\ref{sec:rrnonadaptive}]  & 2-approx [Sec.~\ref{sec:rrnonadaptive}]  \\
\hline
unanimous vote function & $\varphi$-approx [Sec.~\ref{sec:unanimousnonadaptive}] & 2-approx  [Sec.~\ref{sec:unanimousnonadaptive}]
\end{tabularx}

\end{table}
\newpage

\section{Further definitions and background}

A \emph{partial assignment} is a vector $b\in\partassign$. 
We use $f^{b}$ to denote the restriction of function $f(x_1, \ldots, x_n)$ to the bits $i$ with $b_i = *$,
produced by fixing the remaining bits $i$ according to their values $b_i$.
We call $f^b$ the function {\em induced from $f$ by partial assignment $b$}.
We use $N_{0}(b)$ to denote $|\{i|b_i=0\}|$, and $N_{1}(b)$ to denote $|\{i|b_i=1\}|$.

A partial assignment $b'\in\partassign$ is an \emph{extension} of
$b$, written $b'\succeq b$, if $b'_{i} = b_{i}$ for all $i$ such that
$b_{i} \neq *$. We use
$b' \succ b$
to denote that $b' \succeq b$ and $b' \neq b$.



A partial
assignment encodes what information is known at a given
point in a sequential querying (testing) environment. Specifically, for partial assignment $b\in\partassign$, $b_{i}=*$ indicates that
query $i$ has not yet been asked, otherwise $b_i$ equals the
answer to query $i$. We may also refer to query $i$ as {\em test} $i$, and to asking query $i$ as testing or querying bit $x_{i}$, 

Suppose the costs $c_i$ and probabilities $p_i$ for the $n$ queries are
fixed.  We define the expected costs of adaptive evaluation and
verification strategies for $\boolfunc$ or $\classfunc$ as follows.  (The definitions for
non-adaptive strategies are analogous.)
Given an adaptive evaluation strategy ${\cal A}$ for
$f$, and an assignment $x \in \fullassign$, we use
$C({\cal A},x)$ to denote the sum of the costs of the tests performed in using ${\cal A}$ on $x$.  The expected cost of ${\cal A}$ is $\sum_{x \in \fullassign} C({\cal A},x)p(x)$,
where $p(x)=\prod_{i=1}^n p^{x_i}(1-p)^{1-x_i}$.  We say that ${\cal A}$ is an {\em optimal} adaptive
evaluation strategy for $f$ if it has minimum possible expected cost.

Let $L$ denote the range of $f$, and for $\ell in L$, let $\mathcal{X}_{\ell}=\{x\in\fullassign:f(x)=\ell\}$.
An {\em adaptive verification strategy} for $f$ consists of $|L|$ adaptive evaluation strategies ${\cal A}_{\ell}$ for $f$, one for each ${\ell} \in L$.   The expected cost of the verification strategy is 
$\sum_{{\ell} \in L} \left ( \sum_{x \in \mathcal{X}_{\ell}}C({\cal A}_{\ell},x)p(x) \right )$
and it is optimal if it minimizes this expected cost.

If ${\cal A}$ is an evaluation strategy for $f$, we call $\sum_{x \in \mathcal{X}_{\ell}} C({\cal A},x)p(x)$ the $\ell$-cost of ${\cal A}$.
For $\ell \in L$, we say that ${\cal A}$ is {\em $\ell$-optimal} if it has minimum possible $\ell$-cost.
In an optimal verification strategy for $f$, each component evaluation strategy
${\cal A}_{\ell}$ must be $\ell$-optimal.

A Boolean function $\boolfunc$ is symmetric if its output on $x \in \{0,1\}^n$ depends only on $N_1(x)$.
Let $f$ be a symmetric Boolean function $\boolfunc$, or an \emph{unweighted} score classification function $\classfunc$. 
The \emph{value vector} for $f$ 
is the $n+1$ dimensional vector $v^f$, indexed from $0$ to $n$,
whose $j$th entry $v_j^f$ is the value of $f$ on inputs $x$ where $N_1(x) = j$.
We partition value vector $v^{f}$ into
\emph{blocks}. A
block is a maximal subvector of $v^{f}$ such that entries of the
subvector have the same value. 
If $f$ is a score classification function, the blocks correspond to the score intervals,
and block $i$ is the subvector of $v^{f}$ containing the entries in $[\alpha_i, \alpha_{i+1})$.
For $f$ a Boolean function, we define the $\alpha_i$ so that $0=\alpha_1 < \alpha_2 < \leq < \alpha_{B+1}=n+1$ and block $i$ 
is the subvector containing the indices in the interval $[\alpha_i, \alpha_{i+1})$.

We say that assignment $x$ is in the $i$th block if $N_{1}(x)$ is in the interval
$[\alpha_i, \alpha_{i+1})$.
 
With each block $i$ of $v^f$, we associate a function $f^i$, where $f^i(x) = 1$ if $x$ is in block $i$, and $f^i(x)=0$ otherwise.
A \emph{verification strategy for block $i$} is an evaluation strategy for $f^i$.  An \emph{optimal verification strategy for block $i$} is an evaluation strategy for $f^i$ with minimum 1-cost.  

A function $\utilfunc$ is {\em monotone} if $g(b') \geq g(b)$ whenever $b' \succeq b$.  It is {\em submodular} if for $b' \succeq b$, $i$ such that $b'_i = b_i = *$, and $k \in \{0,1\}$, we have $g(b'_{i \leftarrow k}) - g(b') \leq g(b_{i \leftarrow k}) - g(b)$.  Here $b_{i \leftarrow k}$ denotes the partial assignment produced from $b$ by setting $b_i$ to $k$, and similarly for $b'_{i \leftarrow k}$.

\section{Algorithms for the weighted adaptive SSClass problem}
\label{sec:general}
Our first algorithm solves the weighted adaptive SSClass Problem using
the {\em goal value approach} of Deshpande et al.,
a method of designing approximation algorithms for SBFE problems~\cite{DeshpandeGoalValue}.  The approach can easily be extended to the weighted adaptive SSClass problem.
It requires construction of a utility function $\utilfunc$, called a {\em goal function}, associated
with the function $f$ being evaluated.
Function $g$ must be monotone and submodular.  The
maximum value of $g$ must be an integer $Q \geq 0$
such that $g(b) = Q$ iff $f(x)$ 
has the same value for all $x \in \{0,1\}^n$ such that $x\succeq b$. We call $Q$ the \emph{goal value} of
$g$.

An adaptive strategy for evaluating $f$ can then be obtained by applying the Adaptive Greedy algorithm of Golovin and Krause to solve the Stochastic Submodular Cover problem on goal function $g$~\cite{golovinKrause}.  This algorithm greedily chooses the query with highest expected increase in utility, as measured by $g$, per unit cost.  
It follows from the bound of Deshpande et al. on Adaptive Greedy for Stochastic Submodular Cover, that this strategy is an $O(\log Q)$-approximation to the optimal adaptive strategy for evaluating $f$~\cite{DeshpandeGoalValue}.\footnote{Golovin and Krause originally claimed an $O(\log Q)$ bound for Stochastic Submodular Cover~\cite{golovinKrause}, but the proof was recently found to have an error~\cite{nanSaligrama}.  They have since posted a new proof  with an $O(\log^2 Q)$ bound~\cite{golovinKrauseArxivv5}. Deshpande et al.\ proved an $O(\log Q)$ bound using a different proof technique~\cite{DeshpandeGoalValue}.}

We construct $g$ as follows.
Let $r(x)=a_1 x_2 + \ldots + a_nx_n$.  Consider an associated score classification function $f$ defined by 
$\alpha_1, \ldots, \alpha_{B+1}$ and the
$a_i$.
For simplicity, we assume here that the $a_i$ are non-negative.  (The general case is similar.)  
We refer to the values $\alpha_2, \ldots, \alpha_B$ as \emph{cutoffs}.
For each cutoff $\alpha_j$, let $f_j$ denote the Boolean linear threshold function 
$f_j:\{0,1\}^n \rightarrow \{0,1\}$ where $f_j(x) = 1$ if $r(x) \geq \alpha_j$, and $f_j(x)=0$ otherwise.  

Consider a fixed cutoff $\alpha_j$.  Let $\omega = (\sum_i a_i)-\alpha_j+1$. For $b \in \{0,1,*\}^n$, let $r^1(b) = \min\{\alpha_j,\sum_{i:b_i=1} a_i\}$ and $r^0(b) = \min\{\omega,\sum_{i:b_i=0} a_i\}$. 
Note that $r^1(b)=\alpha_j$ iff $f_j(x) = 1$ for all $x \succeq b$, and $r^0(b) = \omega$ iff $f_j(x) = 0$ for all $x \succeq b$.  As shown in~\cite{DeshpandeGoalValue}
the following function $g_j$ is a goal function
for linear threshold function $f_j$, with goal value $\omega\alpha_j$: 
\begin{equation}
  \label{eq:kofnutilfunc}
    g_j(b)= =\omega\alpha_j- (\alpha_j-r^1(b))(\omega-r^0(b)).
\end{equation}



We combine the $B-1$ goal functions $g_j$ using the standard ``AND construction'' for utility functions (cf.~\cite{DeshpandeGoalValue}), which yields a goal
function $g$ for pseudo-Boolean function $f$, where $  g(x) = \sum_{i=1}^{B-1} g_{i}(x).$ Its goal 
  value is at most $(B-1)W^2$ where $W = \sum_i a_i$.






To evaluate $f$, we apply the Adaptive Greedy algorithm to $g$. By the $O(\log Q)$ approximation bound on Adaptive Greedy, this constitutes an algorithm for the adaptive weighted SSClass problem with approximation factor $O(\log BW^2)$, which is $O(\log W)$ since $B \leq W$. In the (unweighted) adaptive SSClass problem, $W=n$, so the approximation factor is $O(\log n)$.

We now describe our simple $B-1$ approximation algorithm for the adaptive unweighted SSClass problem, which takes a very different approach.  It runs the $k$-of-$n$ function evaluation algorithm $B-1$ times, each time setting $k$ to be a different cutoff  $\alpha_j$.  The resulting evaluations are sufficient to determine the correct output class.  The proof that this algorithm achieves a $B-1$ approximation bound is based on the observation that any strategy solving the adaptive SSClass problem is implicitly a strategy for solving each of the $B-1$ induced $k$-of-$n$ problems.  Since we use an optimal algorithm for solving each of those problems, this implies the $B-1$ approximation bound.  
Further details are given in the appendix. When $B$ is small, as for, e.g.,
$k$-of-$n$ functions and the Unanimous Vote function, $B-1$ is a good
approximation. Otherwise, the $O(\log n)$ approximation achieved with the goal value approach may be better.  

By similar arguments, the following is a $3(B-1)$ approximation for the adaptive weighted problem.
For each cutoff $\alpha_j$, use the 3-approximation algorithm of Deshpande et al.\ to evaluate 
linear threshold function $f_j$.  

Combining the above results, we have the following theorem.

\begin{theorem}
  \label{thm:SLSCgoalalg}
  There are two polynomial-time approximation algorithms achieving approximation factors of $O(\log W)$ and $3(B-1)$ respectively for the weighted adaptive SSClass problem.  
There is a polynomial-time algorithm that achieves a $B-1$-approximation for the unweighted adaptive SSClass problem.  
\end{theorem}



\section{Constant-factor approximations for unit-cost problems}
 We begin by reviewing relevant existing techniques.
 
 \subsection{Adaptive Evaluation of k-of-n Functions}

\label{subsec:kofn}
An optimal adaptive strategy, when $f$ is a $k$-of-$n$ function, was given by Salloum, Ben-Dov, and Breuer~\cite{Salloum79,BenDov81,SalloumBreuer84,Chang,Salloum97}. 
The difficulty in finding an optimal strategy is that you do not know a-priori whether the value of $f$ will be 1 or 0.  If 1, then (ignoring cost) it seems it would be better to choose queries with high $p_i$, since you want to get $k$ 1-answers.  Similarly, if 0, it seems it would be better to choose queries with low $p_i$.  The algorithm of Salloum et al.\ is based
on showing that when $f$ is a $k$-of-$n$ function, a 1-optimal strategy is to query the bits in increasing order of $c_i/p_i$ until getting $k$ 1's, while a 0-optimal strategy is to query them in increasing order of $c_i/(1-p_i)$ until getting $n-k+1$ 0's. 

Since the 1-optimal strategy must perform at least the first $k$ tests before terminating, these can be reordered within this strategy without affecting its optimality.  Similarly, the first $n-k+1$ queries of the 0-optimal strategy can be reordered without affecting optimality.
The strategy of Salloum et al.\
is as follows.  If $n=1$, test the one bit.  Else let $S_{1}$ denote the set
of the $k$ bits with smallest $c_{i}/p_{i}$ values.
Let $S_{0}$ denote the set of the $n-k+1$ bits with
smallest $c_{i}/(1-p_{i})$ values.  Since
$|S_{0}|+|S_{1}|=n+1$, by pigeonhole 
$S_{0} \cap S_{1} \neq \emptyset$. Test a bit in $S_0 \cap S_1$. If it is 1,
the problem is reduced to
evaluating the function $f^{1}\colon \bit^{n-1}\to\bit$ where
$f^{1}(x)=1$ iff $N_{1}(x) \geq k-1$. If it is 0,
the problem is reduced to evaluating $f^{0}\colon \bit^{n-1}\to\bit$ where $f^{0}(x)=1$ iff
$N_{1}(x)\geq k$. Recursively evaluate $f^{1}$ or $f^{0}$ as appropriate. Optimality follows from the fact that the chosen bit is an optimal first bit to test in both 0-optimal and 1-optimal strategies.

\subsection{Modified Round Robin}
\label{sec:MRR}
Allen et al.\ \cite{allen2017evaluation} presented a modified
round robin protocol, which is useful in designing non-adaptive strategies when test costs are not all equal.  Suppose that in a sequential testing environment
with $n$ tests, we have $M$ conditions on test outcomes, corresponding
to $M$ predicates on the partial assignments in $\partassign$.
For example, in the $k$-of-$n$ testing problem, we are interested in the following $M=2$ predicates on partial assignments:
(1) 
having at least $k$ ones and (2) having at least $n-k+1$ zeros.
Suppose we are given a testing strategy for each of the $M$
predicates; a strategy stops testing when its predicate is satisfied (by the partial assignment representing
test outcomes), or all tests have been performed.  Let
$\mathrm{Alg}_{1},\dots,\mathrm{Alg}_{M}$ denote those $M$ strategies.
The modified round robin algorithm of Allen et al.\ interleaves
execution of these strategies.  We present a modified version of their algorithm in Algorithm~\ref{alg:roundrobin};  the difference is that their algorithm terminates as soon as one of the
predicates is satisfied, while Algorithm~\ref{alg:roundrobin} terminates when all are satisfied.  

\begin{algorithm}[ht]
\caption{ Modified Round Robin of $M$ Strategies}
\label{alg:roundrobin}
	\begin{algorithmic}
   		\STATE Let $C_{i}\gets 0$ for $i=1,\dots, M$; let $d \gets (*^{n})$
        \WHILE{at least one of the $M$ testing strategies has not terminated}
			\STATE Let $j_{1},\dots,j_{M}$ be the next tests of $\mathrm{Alg}_{1},\dots, \mathrm{Alg}_{M}$ respectively
            \STATE Let  $i^{*} \gets \argmin\limits_{i\in\{1,\dots,M\}} (C_{i}+c_{j_{i}})$
            \STATE Let $t\gets j_{i^{*}}$; let $C_{i^{*}}\gets C_{i^{*}} + c_{t}$
			\STATE Perform test $t$ and set $d_{t}$ to the newly determined value of
  bit $t$
 		\ENDWHILE
	\end{algorithmic}
\end{algorithm}


%
Allen et al.\ showed that
the modified round robin incurs a cost on $x$ that is at most $M$ times the cost incurred by $\mathrm{Alg}_j$ on $x$.
We will use variations on this algorithm and this bound to derive approximation factors for our SSClass problems.

\subsection{A Round Robin Approach to Non-adaptive Evaluation}
\label{sec:rrnonadaptive}
We now present an algorithm for the unit-cost case of the non-adaptive, unweighted SSClass problem.  The pseudocode is presented
in Algorithm~\ref{alg:arbRR}, 
with $\mathrm{Alg}_{1}$ denoting the strategy performing tests in
increasing order of $c_i/p_i$ and 
$\mathrm{Alg}_{0}$ denoting the strategy performing tests in increasing
order of $c_i/(1-p_i)$.  We prove the following theorem.

\begin{algorithm}[ht]
  \caption{Non-adaptive Round Robin Algorithm for SSClass}
  \label{alg:arbRR}
  \begin{algorithmic}
    \STATE Let $C_{0}\gets 0$, $C_{1}\gets 0$
    \STATE Let $d\gets *^{n}$
    \REPEAT
    \STATE Let $j_{0}\gets$ next bit from $\mathrm{Alg}_{0}$
    \STATE Let $j_{1}\gets$ next bit from $\mathrm{Alg}_{1}$
    \STATE Let $j^{*}\gets \argmin_{i\in\bit} C_{i}+c_{j_{i}}$
    \STATE Query bit $i^{*}$ and set $d_{j^{*}}$ to the discovered value
    \UNTIL induced function $f^{d}$ is a constant function 
    \RETURN The constant value of $f^{d}$
  \end{algorithmic}
\end{algorithm}

\begin{theorem}
  \label{thm:arbRR4approx}
  When all tests have unit cost, the expected cost incurred by the non-adaptive Algorithm~\ref{alg:arbRR} is at most 4 times the expected cost of an optimal adaptive strategy for the unweighted adaptive SSClass problem. 
\end{theorem}

By Theorem~\ref{thm:arbRR4approx}, Algorithm~\ref{alg:arbRR} is a 4-approximation for the adaptive \emph{and} non-adaptive versions of the unit-cost unweighted SSClass problem.  
The theorem also implies an upper bound of 4 on the adaptivity gap for this problem.
A simpler analysis 
shows that for arbitrary costs, Algorithm~\ref{alg:arbRR}
achieves an approximation factor of $2(B-1)$ for the non-adaptive version of the problem.
Since the $k$-of-$n$ functions are essentially equivalent to score classification functions with $B=2$, 
the $2(B-1)$-approximation is a 2-approximation for non-adaptive $k$-of-$n$ function evaluation.




\subsection{The Unanimous Vote Function: Adaptive Setting}
\label{sec:unanimousadaptive}
Adaptive evaluation of the Unanimous Vote function
function can be done optimally using the following simple idea.  Recall that 
querying the bits in increasing $c_{i}/p_{i}$ order is optimal for evaluating OR, while querying in increasing $c_{i}/(1-p_{i})$ is optimal for AND.
Now consider the problem of adaptively evaluating the unanimous vote function. Suppose we know the
optimal choice for the first test. 
After the first test, we have an induced SSClass
problem on the remaining bits. If the
first test has value 0, the induced function is equivalent to Boolean OR (mapping UNCERTAIN to 1, and NEGATIVE to 0). 
The subtree rooted at the root node's 0-child should be the 
optimal tree for evaluating OR. Specifically, the remaining bits should be tested in increasing
order of $c_{i}/p_{i}$. 
If, instead, the first bit is 1, the induced function is
equivalent to AND (mapping UNCERTAIN to 0 and POSITIVE to 1) and the remaining bits 
should be queried in increasing
order of $c_{i}/(1-p_{i})$. 

Since we don't actually know the first bit, we can just try each bit as the root and build the rest of the tree according to the optimal OR and AND strategies. We can then calculate the expected cost of each tree, and output the tree with minimum expected cost. 

For succinctness, the optimal OR and AND strategies can be represented by paths, because each performs tests in a fixed order.   
Figure~\ref{fig:optT} shows an example of the strategy computed by the algorithm, where the root is labeled $x_{0}$ and the OR
permutation is the reversal of the AND permutation (which occurs, for
example, with unit costs).



\subsection{A Non-adaptive $\varphi$-approximation for the Unanimous Vote
  Function}
  \label{sec:unanimousnonadaptive}

\begin{figure}[ht]
  \centering
  \includegraphics[scale=0.8]{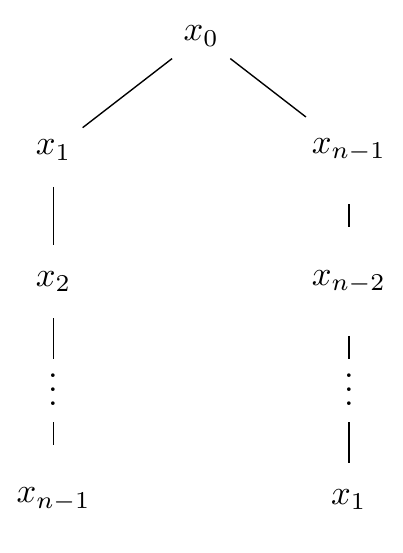}
\caption[The decision tree $T$]{Decision tree $T$ representing optimal
  adaptive strategy with root $x_0$}
\label{fig:optT}
\end{figure}

A simple modification of the round robin makes the algorithm from the previous section non-adaptive, yielding a 2-approximation.
But we now show how to achieve a non-adaptive $\varphi$-approximation in the unit-cost case, where
$\varphi = \frac{1+\sqrt{5}}{2} \approx 1.618$ is the golden ratio. 
We call the algorithm Truncated Round Robin.
We describe the algorithm by describing a subroutine which generates a
permutation of input bits to query, given an initial (root) bit. The
algorithm then tries all possible bits for the root and chooses the
resulting permutation that achieves the lowest expected cost.

Without loss of generality, assume the first bit (the root
node) is $x_{0}$, and the rest are
$x_{1},\dots,x_{n-1}$, and 
$1 > p_{1} \geq p_{2} \geq \dots \geq p_{n-1} > 0$.
Fix $c$ to be a constant such that $0 < c < \frac{1}{2}$.
\begin{algorithm}[ht]
\caption{Truncated Round Robin Subroutine for Unanimous Vote Fn}
\label{alg:NAERR}
	\begin{algorithmic}
        \REQUIRE{$1 > p_{1} \geq p_{2} \geq \dots \geq p_{n-1}$}
   		\STATE Query bit $x_0$
        \STATE Let level $l \gets 1$
        \WHILE {$p_{n-l}<1-c$ \AND $p_l > c$ \AND evaluation unknown}
          \IF {$\left|p_{l}-0.5\right| < \left|p_{n-l}-0.5\right|$}
            \STATE Query $x_{l}$ followed by $x_{n-l}$
          \ELSE \STATE Query $x_{n-l}$ followed by $x_{l}$
          \ENDIF
          \STATE $l \gets l+1$
        \ENDWHILE
        \COMMENT{first phase: alternate branches of tree}
        \WHILE {evaluation unknown}
        	\IF {$p_{l} \geq p_{n-l} \geq 1-c$}
               \STATE Query $x_{n-l}$
            \ELSIF {$c\geq p_{l}\geq p_{n-l}$}
               \STATE Query $x_{l}$
            \ENDIF
            \STATE $l\gets l+1$
        \ENDWHILE
        \COMMENT{second phase: single branch in tree}
	\end{algorithmic}
\end{algorithm}

The subroutine is shown in Algorithm~\ref{alg:NAERR}. 
``Evaluation unknown'' means tests so far were insufficient to determine the output of the Unanimous Vote function. (The output, POSITIVE, NEGATIVE, or UNCERTAIN, is not shown.)

Given 
$x_{0}$ as the root, the optimal adaptive strategy
continues with the OR strategy
(increasing $1/p_{i}$) when $x_{0}=0$, and the AND strategy (increasing $1/(1-p_{i}$)) when $x_{0}=1$. This is shown in
Figure~\ref{fig:optT}, where  $x_{0}=0$ is 
the left branch and  $x_{0}=1$ is the
right.  On the left, we stop querying when we find
a bit with value 1 (or all bits are queried). On the right, we stop when we find a bit with value
0.

Let ``level $l$'' refer to the  tree nodes at distance $l$ from the root; namely, $x_{l}$ and $x_{n-l}$. 
When all costs are 1, the standard round robin technique of the previous section 
in effect tests, for $l=1 \ldots \lceil \frac{n-1}{2}\rceil$,
the bit $x_{l}$ followed by $x_{n-l}$. Note that the algorithm will terminate
by level $\lceil \frac{n-1}{2}\rceil$ because at this point all bits will 
have been queried.

 In the
Truncated Round Robin, we proceed level by level, in two phases. The
first phase concludes once we reach a level $l$ where
$p_{l} > p_{n-l} \geq 1-c$ or $c\geq p_{l} > p_{n-l}$. Let 
$\ell$ denote this level.
In the first phase, we test both $x_{l}$ and $x_{n-l}$, testing first
the variable whose probability is closest to $\frac{1}{2}$. In the
second phase, we abandon the round robin and instead continue down a
single branch in the adaptive tree.
Specifically, in the second phase, if
$p_{l} > p_{n-l} \geq 1-c$, then we continue down the right branch,
testing the remaining variables in increasing order of $p_{i}$. If
$c \geq p_{l} > p_{n-l}$, then we continue down the left branch, 
testing the remaining variables in decreasing order of $p_{i}$.  Fixing $c= \frac{3 - \sqrt{5}}{2}\approx 0.381966$ in the algorithm, the following holds.

\begin{theorem}
  \label{thm:phiapprox}
  When all tests have unit cost, the Truncated Round Robin Algorithm achieves an approximation
  factor of $\varphi$  for non-adaptive evaluation of the Unanimous Vote function.
\end{theorem}

\begin{proof}
Consider the optimal adaptive strategy $T$.  It tests a bit $x_0$ and then follows the optimal AND or OR strategy depending on whether $x_0=1$ or $x_0 = 0$.  Assume the other bits are indexed so $p_1 \geq p_2 \geq \ldots \geq p_{n-1}$.  Thus $T$ is the tree in Figure~\ref{fig:optT}.  Let $C^*_{adapt}$ be the expected cost of $T$.  Let $C^*_{non-adapt}$ be the expected cost of the optimal non-adaptive strategy.  Let $C_{i,TRR}$ be the cost of running the TRR subroutine in (Algorithm~\ref{alg:NAERR}) with root $x_i$.  We use $x_0$ to denote the root of $T$.
Since the TRR algorithm tries all possible roots, its output strategy has expected cost $\min_i C_{i,TRR}$.
We will prove the following claim: $C_{0,TRR} \leq \varphi C^*_{adapt}$. Since the expected cost of the optimal adaptive strategy is bounded above by the expected cost of the
 optimal non-adaptive strategy, the claim implies 
that $\min_i C_{i,TRR} \leq  C_{0,TRR} \leq \varphi C^*_{adapt}$.  Further,
$C^*_{adapt} \leq \varphi C^*_{non-adapt}$, which proves the theorem.

We now prove the claim.
We will write the expected cost of the TRR (with root $x_0$) as
  $C_{0,TRR}=1 + E_{1} + (1-P_{1})E_{2}$.
Here, $E_{1}$ is the expected number of bits tested 
in $T$ in the first phase (i.e.\ in levels $l<\ell$),
$E_{2}$ is the expected number of variables tested among levels
in $T$ in the second phase (levels $l\geq \ell$), given that the second phase is reached,
and $P_{1}$ is the probability of ending during the first phase.  Note that the value of ${\ell}$ is determined only by the values of the $p_i$, and it is independent of the test outcomes.

We will write the expected cost of $T$ (the adaptive tree which is optimal w.r.t\ all
trees with root $x_{0}$) as $C^*_{adapt} = 1 + E'_{1} + (1-P'_{1})E'_{2}$
where $E'_{1}$ is the expected number of bits queried in $T$ before level ${\ell}$, $P'_{1}$ is the probability of ending before level ${\ell}$, and $E'_{2}$ is the expected number of bits queried in levels ${\ell}$ and higher, given that ${\ell}$ was reached.

To prove our claim, we will upper bound the ratio $\alpha := \frac{1 + E_{1} + (1-P_{1})E_{2}}
  {1 + E'_{1} + (1-P'_{1})E'_{2}}.$
Recall that since $c < 1/2$, we have $c < 1-c$.
Also, the first phase ends if all bits have been tested,
which implies that for all $l$ in the first phase, $l \leq \lceil (n-1)/2 \rceil$ so
$p_{n-l} \leq p_{l}$.
We break the first phase
into two parts: (1) The first part consists of all levels $l$ where
$p_{n-l} \leq c < 1-c \leq p_{l}$.
(2) The second part consists of all levels $l$ where $p_{l}\in (c, 1-c)$ or $p_{n-l}\in (c,1-c)$, or both.

Let us rewrite the expected cost $E_{1}$ as $E_{1} = E_{1,1} + (1-P_{1,1})E_{1,2}$.
where $E_{1,1}$ is the expected cost of the first part of phase 1,
$E_{1,2}$ is the expected cost of the second part of phase 1, and
$P_{1,1}$ is the probability of terminating during the first part of
phase 1. Analogously for the cost on tree $T$, we can rewrite
  $E'_{1} = E'_{1,1} + (1-P'_{1,1})E'_{1,2}$.
Then, the ratio we wish to upper bound becomes
  $\alpha = \frac{1 + E_{1,1} + (1-P_{1,1})E_{1,2} + (1-P_{1})E_{2}}
  {1 + E'_{1,1} + (1-P'_{1,1})E'_{1,2} + (1-P'_{1})E'_{2}}$
which we will upper bound by examining the three ratios
  $\theta_1 := \frac{1+E_{1,1}}{1+E'_{1,1}}$,  $\theta_2:=\frac{(1-P_{1,1})E_{1,2}}{(1-P'_{1,1})E'_{1,2}}$ and $\theta_3:=\frac{(1-P_{1})E_{2}}{(1-P'_{1})E'_{2}}$.
  
  For ratio $\theta_1$, notice that the TRR does at most two tests for
every tree level, so $E_{1,1} \leq 2E'_{1,1}$, and thus
$\frac{1+E_{1,1}}{1+E'_{1,1}} \leq
\frac{1+2E'_{1,1}}{1+E'_{1,1}}$. Also,
$\frac{\mathrm{d}}{\mathrm{d}\,x}\left(\frac{1+2x}{1+x}\right) = \frac{1}{(1+x)^{2}} >
0$ for $x>0$. For each path in tree $T$, for
the levels in the first part of the first phase, the probability of
getting a result that causes termination is at least $1-c$. This is
because in the first part,
$p_{l} \geq 1-c > c \geq p_{n-l}$.  If we are taking the left branch (because $x_0=0$) we terminate when we get a test outcome of 1, and on the right ($x_0=1$), we terminate when we get a test outcome of 0. Each bit queried is an independent
Bernoulli trial, so $E'_{1,1} \leq \frac{1}{1-c}$. Because $\frac{1+2x}{1+x}$ is increasing, we can
assert that
  $\theta_1 = \frac{1+E_{1,1}}{1+E'_{1,1}} < \frac{1+2(1-c)^{-1}}{1+1(1-c)^{-1}} = \frac{3-c}{2-c}$.

Next we will upper bound the second ratio $\theta_2$. Let $P(l)$ represent the
probability of reaching level $l$ in the TRR. Further, let $q_{l}$
represent the probability of querying the second bit in level $l$
given that we have reached level $l$. Then, observe that
$(1-P_{1,1})E_{1,2}$ can be written as the sum over all levels $l$ in
phase 1, part 2 of $P(l)(1 + q_{l})$. Note that in phase 1, the first
bit queried is the bit $x_{i}$ such that $p_{i}$ is closest to $0.5$.
Notice also that in the second part of the first phase, each level has
at least one variable $x_{i}$ such that $p_{i}\in (c, 1-c)$. This also
means that $1-p_{i} \in (c, 1-c)$. This means that the first test
performed in any given level in phase 1, part 2 will cause the TRR to
terminate with probability at least $c$. This means that for each
level $l$ in this part of the TRR, we will have $q_{l} \leq 1-c$.

Similarly, $(1-P'_{1,1})E'_{1,2}$ is the sum over all
levels $l$ which comprise phase 1, part 2 in the TRR of $P'(l)$. Here,
$P'(l)$ is defined as the probability of reaching level
$l$ in tree $T$. We do not multiply by $1+q_{l}$ since in the
evaluation of $T$ we only perform one test at each level.

Consider the evaluation of tree $T$ on an assignment. If the
evaluation terminates upon reaching level $l$ in the tree,
for $l < {\ell}$, then the
evaluation using the TRR must terminate at a level $l' \leq l$. That
is, the TRR will terminate at level $l$ \emph{or earlier} for the same
assignment. Thus, we get that $P(l) \leq P'(l)$. Using this, we can
achieve the following bound on the second ratio (letting $S_{2}$
denote the set of all levels included in the second part of phase 1):~
$\theta_2 = \frac{(1-P_{1,1})E_{1,2}}{(1-P'_{1,1})E'_{1,2}} = \frac{\sum_{l\in S_{2}} P(l)(1+q_{l})}{\sum_{l\in S_{2}} P'(l)} \leq \frac{\sum_{l\in S_{2}}P(l)(1+1-c)}{\sum_{l\in S_{2}}P(l)} = 2-c$.

Finally, we wish to upper bound the last ratio,
$\theta_3=\frac{(1-P_{1})E_{2}}{(1-P'_{1})E'_{2}}$. Let $l^{*}={\ell}$ denote the
first level included in the second phase of the TRR. Without loss of
generality, assume that $c \geq p_{l^{*}} \geq p_{n-l^{*}}$ so that in
the TRR, the second phase queries the remaining bits in decreasing
order of $p_{i}$. Thus, all bits $x_{i}$ queried in the second phase
satisfy $p_{i} \leq c$. (The argument is symmetric for the case where
$p_{l^{*}} \geq p_{n-l^{*}} \geq 1-c$).

In this case, any assignments that do not cause termination in the TRR
during the first phase, and that have $x_0 = 0$ (i.e., they would go
down the left branch of $T$), will follow the same path through
the nodes in left branch, for levels ${l^*}$ and higher, that they would have followed in the optimal strategy $T$. (In fact, tests from the
right branch of the tree that were previously performed in phase 1 of the TRR 
do not have to be repeated.)

The numerator of the third ratio $\theta_3$ is equal to the sum, over all
assignments $x$ reaching level $l^*$ in the TRR, of $Pr(x)C_{2}(x)$,
where $C_{2}(x)$ is the total cost of all bits queried in phase 2 for
assignment $x$. Let $Q_0$ be the subset of assignments reaching
level $l^*$ in the TRR which have $x_0 = 0$ and let $Q_1$ be the
subset of assignments reaching level $l^*$ in the TRR which have
$x_0 = 1$. Let $D_{0}$ represent the sum over all assignments in
$Q_{0}$ of $Pr(x)C_{2}(x)$ and let $D_{1}$ represent the sum over all
assignments in $Q_{1}$ of $Pr(x)C_{2}(x)$. Then,
letting $S_{l^{*}}$ represent the set
of assignments reaching level $l^{*}$ in the TRR,
we can rewrite the
numerator of the third ratio as $\sum_{x\in S_{l^{*}}}Pr(x)C_{2}(x) =
  \sum_{x\in Q_{0}}Pr(x)C_{2}(x) + \sum_{x\in Q_{1}}Pr(x)C_{2}(x) =  D_{0} + D_{1}$.
  
  The denominator of the third ratio is the sum, over all assignments
$x$ reaching level $l^*$ in the tree, of $Pr(x)C'_{2}(x)$, where
$C'_{2}(x)$ is the total cost of all bits queried in tree $T$ at
level $l^{*}$ and below. Let $S'_{l^{*}}$ denote the set of
assignments $x$ reaching level $l^{*}$ in tree $T$. Next, observe
that $S_{l^{*}}\subseteq S'_{l^{*}}$ since any assignment that reaches
level $l^{*}$ in the TRR must also reach level $l^{*}$ in the tree. We
can again rewrite the denominator as
  $\sum_{x\in S'_{l^{*}}}Pr(x)C'_{2}(x) \geq \sum_{x\in S_{l^{*}}}Pr(x)C'_{2}(x) = B_{0} + B_{1}$
where $B_{0} = \sum_{x\in Q_{0}}Pr(x)C'_{2}(x)$ and
$B_{1} = \sum_{x\in Q_{1}}Pr(x)C'_{2}(x)$. 
The third ratio $\theta_3$ is thus upper bounded by
  $\frac{(1-P_{1})E_{2}}{(1-P_{1})E_{2}} \leq \frac{D_{0}+D_{1}}{B_{0}+B_{1}}$.
  
   For any $x\in Q_{0}$, the number of bits queried 
in level $l^{*}$ or below in the TRR is less than or equal to the
number of bits queried on $x$ in level $l^{*}$ or below in the
tree.
Thus $D_{0} \leq B_{0}$.

For $x\in Q_{1}$, the number of bits queried at level
$l^{*}$ or below is at least one. Thus $B_{1} \geq J_{1}$, where
$J_{1}$ is the probability that a random assignment $x$ has $x_{0}=1$
and reaches level $l^{*}$.

Note that TRR will terminate on an assignment with $x_{0}=1$ when
it first tests a bit that has value 0. Also note that each bit
$x_{i}$ in level $l^{*}$ and below has probability $p_{i} \leq c$ of
having value 1 and thus probability $1-p_{i} \geq 1-c$ of having value
0 and ending the TRR. Since each bit queried is an independent
trial, the expected number of bits queried before
termination is at most $(1-c)^{-1}$. Thus,
$D_{1} \leq (1-c)^{-1}J_{1}$. Together with the fact that
$D_{0} \leq B_{0}$, we get $
  \frac{D_{0}+D_{1}}{B_{0}+B_{1}} \leq \frac{B_{0}+(1-c)^{-1}J_{1}}{B_{0}+J_{1}}$.
Finally, we observe that since $\frac{B_{0}}{B_{0}} = 1$ and
$\frac{(1-c)^{-1}J_{1}}{J_{1}} \leq \frac{1}{1-c}$, it follows from our
earlier upper bound on $\theta_3$, namely
$\theta_3 \leq \frac{D_{0}+D_{1}}{B_{0}+B_{1}}$, that
  $\theta_3 \leq \frac{D_{0}+D_{1}}{B_{0}+B_{1}} \leq \frac{1}{1-c}$.

Thus, we have three upper bounds:
  (1) $\theta_1 \leq \frac{3-c}{2-c}$, (2)                              $\theta_2 \leq 2-c$, and (3)
$\theta_3 \leq \frac{1}{1-c}$.
  This gives us an upper bound on the ratio of the expected cost of the
TRR to the tree $T$, and thus an upper bound on the approximation
factor. This bound is simply the
maximum of the three upper bounds:
 $ \frac{1 + E_{1} + (1-P_{1})E_{2}}{1 + E'_{1} + (1-P'_{1})E'_{2}} \leq \max\left\{\frac{3-c}{2-c},2-c,\frac{1}{1-c}\right\}
 $.
Setting
$c = \frac{3 - \sqrt{5}}{2} \approx 0.381966$ causes all three upper bounds to equal $\varphi$.
Thus, running the TRR algorithm with $c=\frac{3-\sqrt{5}}{2}$ produces
an expected cost of no more than $\varphi$ times the expected cost of
an optimal strategy.
\end{proof}






\section{Acknowledgements}
The authors were partially supported by NSF Award IIS-1217968. D. Kletenik was also partially supported by a PSC-CUNY Award, jointly funded by The Professional Staff Congress and The City University of New York.
We thank an anonymous referee for suggesting we present our results in terms of SSClass.  We thank Zach Pomerantz for experiments that gave us useful insights into goal functions. 

\bibliography{soda5,throughput,miscrefs}

\appendix
\include{supplementary}

\end{document}

%% file: supplementary.tex


\newcolumntype{Y}{>{\centering\arraybackslash}X}





\section{Verification vs. Evaluation}
\label{append:verification}

 Let $f$ be a symmetric Boolean function.
Let $g$ be the corresponding block identification function.

We use the following terminology, based on ~\cite{Dasetal12}. 

\begin{tabularx}{\linewidth}{l X}
$\mathscr{V}^c(f)$ &  optimal expected verification cost of $f$ with respect to cost vector $c$\\

$\mathscr{C}^c(f)$  &  optimal expected evaluation cost of $f$ with respect to cost vector $c$ \\

$\mathscr{V}^c(g)$ &  optimal expected verification cost of $g$ with respect to cost vector $c$\\

$\mathscr{C}^c(g)$ & optimal expected evaluation cost of $g$ with respect to cost vector $c$\\
\end{tabularx}

It is obvious that ${V}^c(g) \leq \mathscr{V}^c(f) \leq \mathscr{C}^c(f) \leq \mathscr{C}^c(g)$.

Das et al.~\cite{Dasetal12} proved that for symmetric Boolean functions under unit costs, $\mathscr{V}^c(g) = \mathscr{V}^c(f) = \mathscr{C}^c(f) = \mathscr{C}^c(g)$. We show that that does not hold under arbitrary costs. Namely, we show that there exist symmetric Boolean functions for which cost of evaluation exceeds the cost of verification. 

\begin{theorem}
\label{thm:verif}
There exists a symmetric Boolean function $f$ and cost vector $c$ such that $\mathscr{V}^c(g) < \mathscr{C}^c(f)$.
\end{theorem}
\begin{proof}
We give a symmetric function $f$ on $n = 4$ bits that is defined by value vector $v^f = 01100$.  That is, for all $x \in \{0,1\}^n$ with $N_1(x)=j$, then $f(x) = v^f_j$. The blocks of this vector are $B_1=0$, $B_2 = 11$, and $B_3 = 00$. The costs and probabilities for the variables are given in Table~\ref{tab:vars}. 

\renewcommand{\arraystretch}{1}
\begin{table}[h!]
\centering
\begin{tabular}{l l l}
\hline
variable & $p_i$ & cost \\
$x_0$ & 0.1 & 5000 \\
$x_1$ &  0.3 & 6000 \\
$x_2$ & 0.9 & 3000 \\
$x_3$ & 0.8 & 5000 \\
\hline
\end{tabular}
\caption{Table of variables}
\label{tab:vars}
\end{table}

The optimal evaluation tree for $f$ is given in Figure~\ref{fig:opteval}; we denote it as $T$. (Following convention, left edges are implicitly labeled with 0s and right edges with 1s.) It has an expected evaluation cost $C(f) = 14,618$. Note that for any given  root and its left child, the structure of the optimal evaluation tree for $f$  can be determined through a series of $k$-of-$n$ evaluations. Hence, the optimal evaluation tree for $f$ can be found by trying all root-left child combinations and choosing the optimal. Those combinations and the expected tree cost are given in Table~\ref{tab:comb}; the optimal tree cost is bolded.

\begin{figure}[h!]
\centering
\includegraphics[scale=.75]{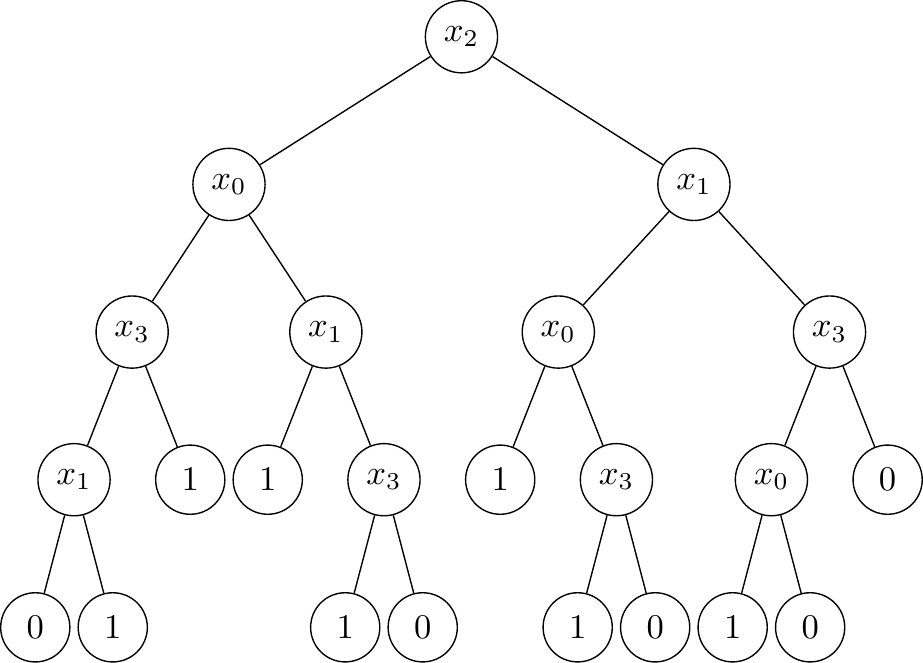}
\caption{Optimal evaluation tree for $f$}
\label{fig:opteval}
\end{figure}

\renewcommand{\arraystretch}{1}

\begin{table}[h!]
\centering
\begin{tabular}{ccc}
\hline
root & left child & expected cost of tree \\

$x_0$ & $x_1$ & 15,529\\
$x_0$ & $x_2$ & 15,259 \\
$x_0$ & $x_3$ & 16,042 \\
$x_1$ & $x_0$ & 14,881 \\
$x_1$ & $x_2$ & 14,643 \\
$x_1$  & $x_3$& 15,616\\
$\mathbf{x_2}$ & $\mathbf{x_0}$ & \textbf{14,618} \\
$x_2$ & $x_1$ & 14,670\\
$x_2$ & $x_3$ & 14,623\\
$x_3$ & $x_0$ &  15,394\\
$x_3$ & $x_1$ & 15,616\\
$x_3$ & $x_2$ & 15,406\\
\hline
\end{tabular}
\caption{Possible trees for $f$ and their cost}
\label{tab:comb}
\end{table}

The expected cost of verifying that an assignment is in $B_2$ using $T$ is $10,248.8$.

But the optimal verification cost for $B_2$ is actually $10,241.8$. That cost is achieved in the tree in Figure~\ref{fig:optverif}. (The leaf nodes labeled X are nodes that the verification tree can never reach; they correspond to assignments  not in $B_2$.) Hence, $\mathscr{C}^c(g) \neq \mathscr{V}^c(g)$.

\begin{figure}[h!]
\centering
\includegraphics[scale=.75]{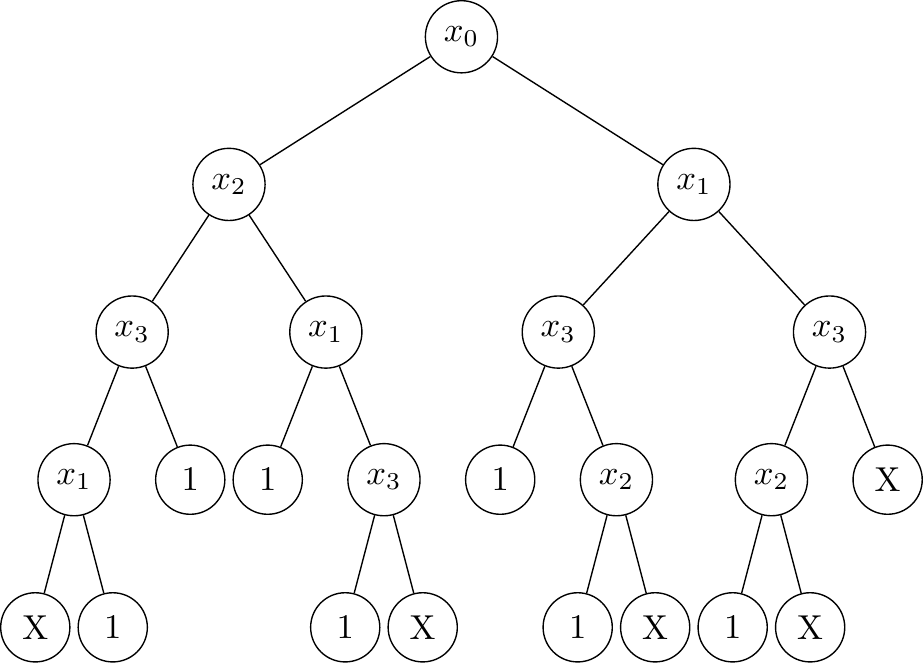}
\caption{Optimal verification tree for $B_2$}
\label{fig:optverif}
\end{figure}

The construction of this counterexample was based on the following observations.
The optimal verification tree for $B_1$ is obvious since it must test all four variables on
assignments in $B_1$ (in any order). The optimal verification tree for $B_3$ is obvious as well; since it must verify the block by finding at least three 1's, it tests the variables in increasing order of $\frac{c_i}{p_i}$ and terminates as soon as three 1's are found. However, since at least three variables must be tested, any tree that tests the three cheapest $\frac{c_i}{p_i}$ variables first, in any order, has the same (optimal) cost. We call the set of all trees that test those variables first $S\{T_{B_3}\}$; it is the set of all optimal verification trees for $B_3$,

The optimal verification tree for $B_2$ is less obvious; however, given variables for the root and its left child, the rest of the tree follows from a series of $k$-of-$n$ evaluations, just like $T$. We give the structure for the tree in Figure~\ref{fig:optstruct} and denote it is as $T_{B_2}$. 

\begin{figure}
\centering
\includegraphics[scale=.75]{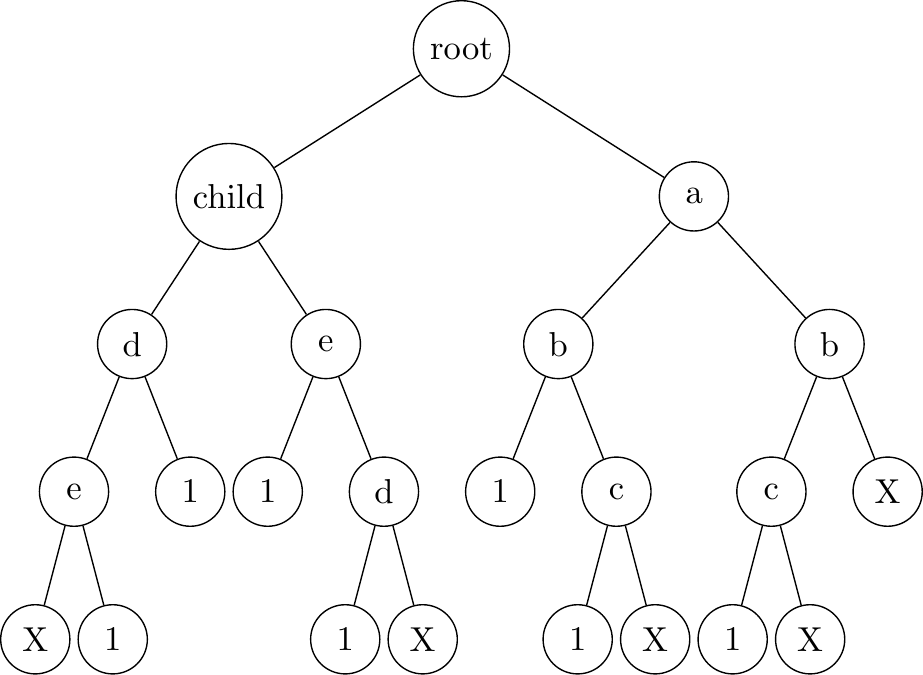}
\caption{$T_{B_2}$: Optimal verification tree structure for $B_2$}
\label{fig:optstruct}
\end{figure}

Specifically, the rules for the nodes of $T_{B_2}$ are as follows:

Nodes a, b, and c are the remaining variables on the right-hand side after the root is chosen, ordered in increasing order of $\frac{c_i}{1-p_i}$. This is due to the fact that once the root node is tested and has the value of 1, the goal is to find 0's as cheaply as possible. 

Node d is chosen to be the variable with the maximum $\frac{c_i}{p_i}$; since two 0's have already been found, the goal is to find cheap 1's. 

Finally, node e is again chosen to be the variable with low $\frac{c_i}{1-pi}$, reflecting once again that once one 1 has been found, the goal is to find cheap 0's.

If the root of the optimal verification tree for $B_2$ has the maximum value for $\frac{c_i}{p_i}$, and furthermore, the variable tested in node \emph{child} differs from the variables tested in a and b, $T_{B_2}$ will differ from all of the trees in $S\{T_{B_3}\}$.  This is in fact the case for $f$ under the cost and probabilities given in Table~\ref{tab:vars}. Hence the optimal evaluation tree for function $f$, $T$, must achieve a non-optimal verification cost on either block $B_2$ or $B_3$.

(We note that the particular variables given here are far from the only choice of variables that satisfy these conditions and prove the theorem. They were chosen as an illustrative example.)
\end{proof}

A corollary follows:
\begin{corollary}
  \label{cor:verifvseval}
For all interval functions $f$ and cost vectors $c$, $\mathscr{V}^c(g) = \mathscr{V}^c(f)$.
\end{corollary}
\begin{proof}
For the particular function $f$ given above, defined by value vector 01100, verifying the value of the function when it is 1 is equivalent to verifying block $B_2.$  Verifying the value of the function when it is 0 requires verifying either block $B_1$ or $B_3$; however, since the optimal verification strategy for $B_1$ is to test every bit (in any order), the optimal verification tree for $B_3$ is the optimal verification tree for $f=0$. Hence, $\mathscr{V}^c(f) = \mathscr{V}^c(g)$ for any cost vector $c$.

More generally, for any three-blocked value vector, the verification tree for the value of the function will either be the verification tree for the middle block, or a verification tree for blocks 1 and 3. Whenever it is the latter, there will always exist at least one bit in the intersection of the optimal verification strategies for blocks 1 and 3. Then we can use a strategy similar to the one in Section~\ref{subsec:kofn} to continuously choose the bit in the intersection of the strategies to form the optimal verification tree.
In doing so, we replace the verification trees for the first and last blocks with a tree of equal expected cost. 

Hence, for any three-block value vector, $\mathscr{V}^c(g) = \mathscr{V}^c(f)$.
\end{proof}

\section{Background: Optimality of the k-of-n Algorithm}
\label{append:kofn}
In Section~\ref{subsec:kofn}, we described the known algorithm for evaluating $k$-of-$n$ functions~\cite{Salloum79,BenDov81,SalloumBreuer84,Chang,Salloum97}.  
It is helpful to understand why this algorithm is, in fact, an optimal adaptive evaluation strategy.  Here we review a version of the proof that is given in~\cite{unluyurtBorosDoubleRegular}.

The proof relies on the fact that evaluating the bits in increasing $c_i/p_i$ order is a 1-optimal strategy, and evaluating them in increasing $c_i/(1-p_i)$ ordering is a 0-optimal strategy.  (We omit the proof of this fact here.)  Thus these two strategies constitute an optimal verification strategy.

The expected cost of this optimal verification strategy is a lower bound on the expected cost of an optimal evaluation strategy.  If $f(x)=1$, the 1-optimal strategy cannot terminate on $x$ before it has tested all $k$ bits in $S_1$.  Thus the strategy is still 1-optimal if those bits are permuted.  Similarly, if $f(x)=0$, the 0-optimal strategy cannot terminate before it has tested all bits in $S_0$, and those can be permuted.  If $x_i \in S_1 \cap S_0$, it there is both a 1-optimal strategy and a 0-optimal strategy that tests $x_i$ first.
Inductively, it follows that the above $k$-of-$n$ evaluation strategy is both 1-optimal and 0-optimal.  Since its expected cost is equal to the optimal expected verification cost, it is an optimal evaluation strategy.  

\section{Omitted Proofs and Related Material}
\label{append:proofs}

\subsection{Details of the $B-1$ approximation}
\label{sec:B-1approx}

%
Let $f:\{0,1\}^n \rightarrow \{1, \ldots, B\}$ be 
the unweighted score classification function associated with the values
$0 = \alpha_1 < \ldots < \alpha_B < \alpha_{B+1}=n+1$.
Let $v=v^f$ be its value vector. 
An assignment $x$ belongs to block $j$ if $\alpha_j \leq N_1(x) < \alpha_{j+1}$.

We present Algorithm~\ref{alg:repeatedkofn} and
show it achieves a $(B-1)$-approximation for the Symmetric SLSC problem.
In the algorithm, we denote as $f_{i}$ the $k$-of-$n$
function with $k=\alpha_i$. We note that in different iterations of the for loop, the strategy that is executed in the body may choose a test that was already performed in a previous iteration. 
The test does not actually have to be repeated, as the outcome can be stored after the first time the test is performed, and accessed whenever the test is chosen again.  


\begin{algorithm}[ht]
  \caption{Adaptive Algorithm for Evaluating Score Classification Functio $f$}
  \label{alg:repeatedkofn}
  \begin{algorithmic}
    \FOR{$i\gets 2$ {\bfseries to} $B$}
    \STATE Run the optimal adaptive $k$-of-$n$ strategy to evaluate $f_{i}(x)$
    \ENDFOR
    \STATE Let $i^{*}\gets \max\{i\mid f_{i}(x)=1\}$~~~//~$i^{*}=\alpha_1 = 0$ if $f_i(x)=0$ \\~~~~~~~~~~~~~~~~~~~~~~~~~~~~~~~~~~~~~~~~~~~~~~~~~~~~~~~~~for all $i > 1$
     \RETURN $v_{\alpha_{i^{*}}}$
  \end{algorithmic}
\end{algorithm}

The  correctness of the algorithm follows easily
from the fact that
$\alpha_{i^{*}}\leq N_{1}(x) < \alpha_{i^{*}+1}$, and so
$f(x) = v^{f}_{\alpha_{i^{*}}}$.

We now examine the expected cost of the strategy computed in
Algorithm~\ref{alg:repeatedkofn}. Let $C(f_{i})$ denote the expected cost of
evaluating $f_{i}$ using the optimal $k$-of-$n$ strategy. Let $\OPT$ be expected cost of the optimal adaptive strategy for $f$.

\begin{lemma}
  \label{lem:kofnltopt}
  $C(f_{i})\leq \OPT$ for $i \in \{1, \ldots, B-1\}$.
\end{lemma}

\begin{proof}
Let $T$ be an optimal adaptive strategy for evaluating $f$.
Consider using $T$ to evaluate $f$ on an initially unknown input $x$.
When a leaf of $T$ is reached, we have discovered the values of some of the
bits of $x$.  Let $d$ be the partial assignment representing that knowledge.
Recall that $f^{d}$ is the function induced from $f$ by $d$.
The value vector of $f^{d}$ is a subvector of $v^f$, the value vector of $f$.
More particularly, it is the subvector stretching from index $N_1(d)$ of $v^f$
to index $n - N_0(d)$.  
Since $T$ is an evaluation strategy for $f$,
reaching a leaf of $T$ means that we have enough information to determine $f(x)$.
Thus all entries of the subvector must be equal, implying that it is contained within a single
block of $v^f$.  We call this the block associated with the leaf.  

For each block $i$, we can create
a new tree $T'_{i}$ from $T$ which evaluates the
  function $f_{i}$.  We do this by relabeling the leaves of $T$: if the leaf is associated with block $i'$, then we label the leaf with output value 1 if $i' > i$, and with 0 otherwise.   $T'_{i}$ is an adaptive strategy for evaluating $f_{i}$.

  The expected cost of evaluating $f_{i}$ 
  using $T'_{i}$ is equal
  to $\OPT$, since the structure of the tree is unchanged from $T$
  (we've only changed the labels). Since $T'_{i}$ cannot do better
  than the optimal $k$-of-$n$ strategy, $C(f_{i}) \leq \OPT$.
\end{proof}

This yields an approximation bound for Algorithm~\ref{alg:repeatedkofn}.

\begin{theorem} Algorithm~\ref{alg:repeatedkofn} is a $(B-1)$-approximation algorithm for the unweighted adaptive SSClass problem.
\end{theorem}
\begin{proof}
  The total cost incurred
  by the algorithm is no greater than the sum of the costs incurred by the $B-1$ runs of the $k$-of-$n$ algorithm. Thus by Lemma~\ref{lem:kofnltopt}, $\ALG \leq \sum_{i=1}^{B-1} C(f_{i}) \leq \sum_{i=1}^{B-1}\OPT$.
\end{proof}


\subsection{The $2(B-1)$ approximation for non-adaptive unweighted SSClass, arbitrary costs}
\label{sec:nonadaptive2B-1}

We briefly mentioned the result in Section~\ref{sec:4approx}. Note that we already have a simple $B-1$ approximation algorithm for the adaptive case.  

\begin{theorem}
  \label{thm:2Bminus1approx}
  Algorithm~\ref{alg:arbRR} is a $2(B-1)$-approximation for the non-adapative unweighted SSCLass problem.
\end{theorem}

\begin{proof}
Let $f:\{0,1\}^n \rightarrow \{1, \ldots, B\}$ be the score classification function associated with an instance of the problem.  Let $\mathcal{A}$ be an optimal non-adaptive algorithm for evaluating $f$ and let $OPT$ be its expected cost.

Consider running Algorithm~\ref{alg:arbRR} to evaluate $f$.
  For each assignment $a\in\fullassign$, there is some block boundary
  $\alpha_{i}$ that is the final block boundary ``crossed'' before
  execution of Algorithm~\ref{alg:arbRR} terminates. In other words, immediately 
  before the final test is chosen, the value vector of the 
  pseudo-Boolean function induced by the prior test results contains
  entries $\alpha_i-1$ and  $\alpha_i$ of the original value vector, where $i$ is the index of a block of that vector.  The final test will cause the induced value vector to contain only one of these entries, thereby determining whether $x$ is in block $i-1$ or block $i$.  
 Either way, we say that
$\alpha_{i}$ was the final block boundary crossed.

  There are $B-1$ possible final block boundaries, $\alpha_2, \ldots, \alpha_{B}$. 
We will partition the assignments $x \in \{0,1\}^n$ into sets $S_{i}$
  for $i\in\{2,3,\dots,B\}$ where each set $S_{i}$ contains all
  assignments on which execution of Algorithm~\ref{alg:arbRR}
  terminates after crossing block boundary $\alpha_{i}$. 
  Let $RR$ denote the strategy of Algorithm~\ref{alg:arbRR}.

Quantity $C(\mathrm{RR}, a)$ is the cost incurred by $RR$ on assignment $a$.
  Let $C^{RR}(\mathrm{Alg}_{0},a)$ and $C^{RR}(\mathrm{Alg}_{1},a)$ represent
  the cost incurred on assignment $a$ during the execution of $RR$ by $\mathrm{Alg}_{0}$ and $\mathrm{Alg}_{1}$
  respectively, so $C(\mathrm{RR},a) =  C^{RR}(\mathrm{Alg}_{0},a) + C^{RR}(\mathrm{Alg}_{1},a)$.
  Let $Q_{0}$ and $Q_{1}$ be the sets of
  assignments $a$ for which the final bit queried in Algorithm~\ref{alg:arbRR} was determined by
  $\mathrm{Alg}_{0}$ and $\mathrm{Alg}_{1}$, respectively.
  
  Let $f_i$ denote the $k$-of-$n$ function with $k = \alpha_i$.
  Let $\mathrm{Alg}^i_{0}$ denote the 0-optimal strategy for evaluating $f_i$, which queries bits in increasing order of $c_i/(1-p_i)$ until $n-\alpha_i+1$ 0's are obtained, or all bits are queried. Similarly, let $\mathrm{Alg}^i_{1}$ denote the 1-optimal strategy for evaluating $f_i$, which queries bits $j$ in increasing order of $c_j/p_j$ until $\alpha_i$ 1's are obtained, or all bits are queried. 
  We have the following two inequalities, one each for $Q_{0}$ and $Q_{1}$.
  \begin{gather}
    \label{eq:arbRRbbsum1}
    \sum_{a\in S_{i}\cap Q_{1}} C(\mathrm{RR},a)p(a) \\
    \leq \sum_{a\in S_{i}\cap Q_{1}}2C^{RR}(\mathrm{Alg}_{1},a)p(a) \\
    \leq \sum_{\substack{a\in\fullassign \\ N_{1}(a)\geq \alpha_i}} 2C(\mathrm{Alg}^i_{1},a)p(a)
    \end{gather}
    \begin{gather}
    \label{eq:arbRRbbsum0}
    \sum_{a\in S_{i}\cap Q_{0}} C(\mathrm{RR},a)p(a) \\
    \leq \sum_{a\in S_{i}\cap Q_{0}}2C^{RR}(\mathrm{Alg}_{0},a)p(a) \\
    \leq \sum_{\substack{a\in\fullassign \\ N_{1}(a)\geq n-\alpha_i+1}} 2C(\mathrm{Alg}^i_{0},a)p(a)
  \end{gather}
  For each, the first inequality holds because
  $C(\mathrm{RR}, a) = C^{RR}(\mathrm{Alg}_{0},a) + C^{RR}(\mathrm{Alg}_{1},a)$. Further,
  it holds that $C^{RR}(\mathrm{Alg}_{0},a)\leq C^{RR}(\mathrm{Alg}_{1},a)$ for
  assignments in $Q_{1}$ (and similarly for assignments in
  $Q_{0}$). 
  
  As in the proof of Lemma~\ref{lem:kofnltopt}, the strategy $\mathcal{A}$ for evaluating $f$ 
  could be turned into a strategy for evaluating $f_i$ by relabeling the leaves of $\mathcal{A}$,
  without changing the cost incurred by the strategy on any assignment.  Since $\mathrm{Alg}^i_{0}$ is a 0-optimal strategy for $f^i$, $\mathrm{Alg}^i_{1}$ is a 1-optimal strategy for $f^i$, $f^i(a) = 1$ iff $N_1(a) \geq b_i$, and $f^i(a) = 0$ iff $N_1(a) < \alpha_i$ (equivalently, $N_0(a) \geq n-\alpha_i + 1$),

  \[
  	\sum_{\substack{a\in\fullassign \\ N_{1}(a)\geq \alpha_i}} C^{RR}(\mathrm{Alg}_{1},a)p(a) 
    \leq \sum_{\substack{a\in\fullassign \\ N_{1}(a)\geq \alpha_i}} C(\mathcal{A},a)p(a)
  \]
  and
  \[
    \sum_{\substack{a\in\fullassign \\ N_{1}(a)< \alpha_i}} C^{RR}(\mathrm{Alg}_{0},a)p(a) 
    \leq \sum_{\substack{a\in\fullassign \\ N_{1}(a)< \alpha_i}} C(\mathcal{A},a)p(a)\rlap{\,.}
  \]
  
  Using this, we sum the two quantities of \eqref{eq:arbRRbbsum1} 
  and \eqref{eq:arbRRbbsum0} to get the following inequality representing the cost 
  incurred for assignments in $S_{i}$.
  \begin{equation}
    \begin{split}
    \label{eq:arbRRSisum}
    \sum_{a\in S_{i}} C(\mathrm{RR}, a)p(a)
    &\leq 2
      \sum_{\substack{a\in\fullassign \\ N_{1}(a)\geq \alpha_i}} C^{RR}(\mathrm{Alg}_{1},a)p(a) \\
      &+ 2\sum_{\substack{a\in\fullassign \\ N_{1}(a)<\alpha_i}} C^{RR}(\mathrm{Alg}_{0},a)p(a)
    \\
    &\leq 2\sum_{a\in\fullassign}C(\mathcal{A}, a)p(a) = 2\OPT
  \end{split}
  \end{equation}
  Summing over all block boundaries we get
  \begin{equation}
   \label{eq:arbRRexpcost}
   \begin{split}
    \sum_{a\in\fullassign}C(\mathrm{RR},a)p(a)
   & = \sum_{i=1}^{B-1}\sum_{a\in S_{i}}C(\mathrm{RR},a)p(a)
    \leq 2(B-1)\OPT
   \end{split}
  \end{equation}
as desired.
\end{proof}

\subsection{Proof of the 4-approximation for Unweighted SSClass with Unit Costs}
\label{sec:4approx}
Before proving Theorem~\ref{thm:arbRR4approx}, we first prove
some claims.
Consider applying Algorithm~\ref{alg:arbRR} to evaluate the pseudo-Boolean function $f$ associated with a symmetric SLSC function $f$.  Assume further that the costs $c_i$ are all equal to 1.
Let $\beta_j = \alpha_{j+1}$.
Consider block $j$ of $v^f$, represented by $[\alpha_j,\beta_j)$.
Let $M^j=\{a \in \fullassign \mid \alpha_j \leq N_1(a) < \beta_j \}$.
That is, $M^j$ is the set of assignments in the $j$th block.

For a permutation $\sigma$ and an assignment $a \in M^j$,
let $c_{1}^j(\sigma,a)$ denote the
total cost incurred when bits are queried in the order specified by $\sigma$, until it is verified that
$N_{1}(a)\geq \alpha_j$ (i.e., until $\alpha_j$ 1's are seen).
Similarly, let
$c_{0}^j(\sigma,a)$ denote the total cost incurred until it is verified that
$N_{1}(a) < \beta_j$ (equivalently,
$n-\beta_j+1$ 0's are seen).  Since we are assuming unit cost tests, total cost incurred is equal to the number of bits queried.

Let $C_{1}^j(\sigma) = \sum_{a\in M^j}[c_{1}^j(\sigma,a)p(a)]$ and
similarly
$C_{0}^j(\sigma) = \sum_{a\in M^j}[c_{0}^j(\sigma,a)p(a)]$.

Let $\sigma^1$ be the permutation that orders bits in increasing order of $1/p_i$ (equivalently, decreasing order of $p_i$), and let $\sigma^0$ be the permutation that orders bits in increasing order of $1/(1-p_i)$ (equivalently, increasing order of $p_i$).
For simplicity, we assume in what follows that the $p_i$ are all different; the arguments can be easily extended if this is not the case.

\begin{claim}
  \label{clm:prob0prob1opt}
  $C^j_{1}(\sigma^1) \leq C^j_{1}(\sigma)$ for all permutations $\sigma$.
Similarly,
  $C^j_{0}(\sigma^0) \leq C^j_{0}(\sigma)$ for all permutations $\sigma$.
\end{claim}

\begin{proof}
We give the proof for $C^j_{1}$.  The proof for $C^j_{0}$ is analogous.

Suppose there exists a permutation
$\pi$ such that $C^j_{1}(\pi) < C^j_{1}(\sigma^1)$.
Let $\pi$ be an optimal such permutation, so
$C^j_{1}(\pi) \leq C^j_{1}(\sigma)$ for all permutations $\sigma$.
Renumber the bits so that $\pi(i) = i$ for all $i$.

Since the $p_i$'s are distinct and $\pi \neq \sigma^1$,
there exists a bit
$1\leq l\leq n-1$, such that $p_l < p_{l+1}$.
Consider the permutation $\pi'$ produced from $\pi$ by swapping the elements in positions $l$ and $l+1$.

We will obtain a contradiction by showing that
$C^j_{1}(\pi') < C^j_{1}(\pi)$. Consider the four possible values of
$x_l$ and $x_l+1$:
\begin{itemize}
\item $x_l=0$ and $x_{l+1}=0$
\item $x_{l}=1$ and $x_{l+1}=1$
\item $x_{l}=0$ and $x_{l+1}=1$
\item $x_{l}=1$ and $x_{l+1}=0$
\end{itemize}
Consider the difference
\begin{equation*}
  C^j_{1}(\pi) - C^j_{1}(\pi') =
  \sum_{a\in M^j}\left[c_{1}^j(\pi,a)-c_{1}^j(\pi',a)\right]p(a)
\end{equation*}
and consider a specific assignment, $a \in M^j$.
Let $d$ represent the partial
assignment where $d_{i}=a_{i}$ for all $i$ such that
$i<l$ and $d_{i}=*$ otherwise. That is, $d$ contains the values of the
variables which appear before $x_{l}$ in permutation $\pi$ (and
before $x_{l+1}$ in permutation $\pi'$).

If $N_{1}(d) < \alpha_j-1$, then verifying $N_1(a) \geq \alpha_j$ using $\pi$ results in querying both $x_{l}$ and $x_{l+1}$,
so
$c_{1}^j(\pi,a)=c_{1}^j(\pi',a)$.  If
$N_{1}(d)\geq \alpha_j$, then verifying $N_1(a) \geq \alpha_j$ using $\pi$ does not involve querying either $x_{l}$ or $x_{l+1}$,
so $c_{1}^j(\pi,a)=c_{1}^j(\pi',a)$.

Suppose $N_{1}(d)=\alpha_j-1$. In
this case, if $a_{l} = a_{l+1}=0$, then $\pi$ and
$\pi'$ will query both $x_{l}$ and $x_{l+1}$ and incur the same total cost. If
$a_{l} = a_{l+1}=1$, then $\pi$ and $\pi'$ will each
query exactly one of $x_{l}$ and $x_{l+1}$ before terminating.
Since both queries have unit cost,
$\pi$ and $\pi'$ will incur the same total cost on $a$.

We are left with the assignments $a \in M_j$ where for the corresponding $d$,
$N_{1}(d)=\alpha_j-1$ \emph{and} $a_{l}\neq a_{l+1}$. Let $A$
represent the set of such assignments. It follows that
\[
  C^j_{1}(\pi) - C^j_{1}(\pi') =
  \sum_{a\in A}\left[c_{1}^j(\pi,a)-c_{1}^j(\pi',a)\right]p(a)
\]
Let $p(a,i) = (p_i)^{a_i} (1-p_i)^{(1-a_i)}$.
Then $p(a) = \prod_{i =1}^n p(a,i)$.
Let $p'(a) = p(a)/[p(a,l) \cdot p(a,l+1)]$.
Observe that
for $a\in A$, both permutation $\pi$ and $\pi'$ will result in terminating
after querying $l$ or $l+1$ bits (which of the two
depends on the values of $a_l$ and $a_{l+1}$).
There are two
cases to consider:
\begin{enumerate}
\item $a_{l}=1$ and $a_{l+1}=0$. In this case,
  $c_{1}^j(\pi,a)=l$ and $c_{1}^j(\pi',a)=l+1$.
  $p(a)=p'(a)\cdot p_{l}(1-p_{l+1})$.
\item $a_{l}=0$ and $a_{l+1}=1$. In this case,
  $c_{1}^j(\pi,a)=l+1$ and $c_{1}^j(\pi',a)=l$.
  $p(a)=p'(a)\cdot (1-p_l)p_{l+1}$.
\end{enumerate}
In the first case, we get
\begin{equation*}
  \begin{split}
  \left[c_{1}^j(\pi,a) - c_{1}^j(\pi',a)\right]p(a)
  &= \left[ p'(a)\cdot p_{l}(1-p_{l+1})\right]\left[l - (l+1)\right] \\
  &= - p'(a)\cdot p_{l}(1-p_{l+1})
\end{split}
\end{equation*}
and in the second case, we get
\[
  \left[c_{1}^j(\pi,a) - c_{1}^j(\pi',a)\right]p(a) =
  p'(a)\cdot (1-p_{l})p_{l+1}.
\]
Let $Q_{10}$ represent the set of assignments which fall in the first
case, and $Q_{01}$ the set of assignments which fall in the second
case. Note that each assignment $a \in Q_{10}$ has a corresponding assignment
$\hat{a}$ in $Q_{01}$ which is identical except in bits $l$ and $l+1$.
Further, $p'(a) = p'(\hat{a})$.  Thus

\begin{equation}
\begin{split}
  &C^j_{1}(\pi) - C^j_{1}(\pi') \\
  &= \sum_{a\in Q_{01}} p'(a)\cdot (1-p_{l})p_{l+1}
  - \sum_{a\in Q_{10}} p'(a)\cdot p_{l}(1-p_{l+1}) \\
  &= \sum_{a\in Q_{01}} [p'(a)\cdot (1-p_{l})p_{l+1} - p'(\hat{a})\cdot p_{l}(1-p_{l+1})] \\
  &= \sum_{a\in Q_{01}} p'(a)(p_{l+1} - p_{l})~~~\mbox{ since } p'(a) = p'(\hat{a}) \\
  &= (p_{l+1} - p_{l})\sum_{a\in Q_{01}} p'(a)
\end{split}
\end{equation}
But since
$p_{l} < p_{l+1}$ and the $p'(a)$ are non-negative,
$C^j_{1}(\pi) > C^j_{1}(\pi')$. This contradicts the optimality of $\pi$.

A symmetric argument shows that $\sigma^0$ minimizes $C^j_0$.
\end{proof}

Let $T$ be a decision tree representing an adaptive testing strategy. For an assignment $a\in M^j$, let $c_1^j(T,a)$ denote the total cost incurred when bits are queried as specified by $T$, until it is verified that $N_1(a)\geq \alpha_j$. Similarly, define $c_0^j(T,a)$ as the total cost incurred by the adaptive strategy $T$ when querying bits until it is verified that $N_1(a) < \beta_j$. We similarly define $C_1^j(T)=\sum_{a\in M^j} [c_1^j(T,a)p(a)]$ and $C_0^j(T) = \sum_{a\in M^j} [c_0^j(T,a)p(a)]$. We can further claim that not only are $\sigma^1$ and $\sigma^0$ better than any other permutation (in terms of $C_1^j$ and $C_0^j$) but also that they are optimal with respect to adaptive strategies. That is:
\begin{remark}
  For any $j$, and for all adaptive strategies $T$: $C_1^j(\sigma^1) \leq C_1^j(T)$ and $C_0^j(\sigma^0)\leq C_1^j(T)$.
\end{remark}
\begin{proof}
We will prove this by arguing that the optimal adaptive strategy (with respect to $C_1^j$ or $C_0^j$) is in fact a permutation (i.e., is nonadaptive). Then, it must follow from Claim~\ref{clm:prob0prob1opt} that this adaptive strategy is $\sigma^1$ (respectively, $\sigma^0$).

We do this by induction on $n$. For $n=1$, the optimal adaptive strategy is to query the single bit. Then, assume that for any function on $n$ bits, the adaptive strategy which minimizes $C_1^j$ (resp.\ $C_0^j$) is the permutation $\sigma^1$ (resp.\ $\sigma^0$). Then, for a function on $n+1$ bits, the optimal adaptive strategy is a decision tree with some bit at the root. Whether this first bit is a 0 or a 1, the result induces a new function on $n$ variables (the same $n$ variables for either outcome), and the optimal strategy in this case is the permutation that orders bits by increasing order of $1/p_i$ (resp.\ $1/(1-p_i)$). Thus the subtrees rooted at the 0-child and 1-child of the root are in fact the same permutation, and thus the entire strategy can be expressed as a permutation of the $n+1$ bits: Choose the root first, then go in increasing order of $1/p_i$ (resp.\ $1/(1-p_i)$). Since the strategy minimizing $C_1^j$ (resp.\ $C_0^j$) for $n+1$ bits is a permutation, by Claim~\ref{clm:prob0prob1opt}, it must be the permutation $\sigma^1$ (resp.\ $\sigma^0$).
\end{proof}


For a strategy $A$ and assignment $a$, let $C(A,a)$ denote the cost
incurred evaluating $a$ using strategy $A$. Thus, the expected cost of
strategy $A$ is $\sum_{a\in\fullassign}C(A,a)p(a)$.

Now let $\mathbb{A}_{OPT}$ be an adaptive strategy that minimizes the
expected cost of evaluating $f$.  Let $T_{OPT}$ be the corresponding
decision tree of this adaptive strategy.

\begin{claim}
  \label{clm:prob01ltopt}
  $C^j_{0}(\sigma^0) \leq \sum_{ a \in M^j} C(\mathbb{A}_{OPT},a)p(a)$ and
  $C^j_{1}(\sigma^1) \leq \sum_{ a \in M^j} C(\mathbb{A}_{OPT},a)p(a)$.
\end{claim}

\begin{proof}
  In evaluating $f$ on some input $a\in M^{j}$, we cannot terminate
  until we have seen at least $\alpha_j$ ones \emph{and} at least
  $n - \beta_j+1$ zeros.  Thus if we perform tests on $a$ in the order
  indicated by $T_{OPT}$, and terminate as soon as we see $\alpha_j$
  ones, the resulting cost will be at most $C(\mathbb{A}_{OPT},a)$.
  Thus
  $\sum_{a \in M^j} c_1^j(T_{OPT},a)p(a) \leq \sum_{a \in M^j}
  C(\mathbb{A}_{OPT}, a)p(a)$.  Since $\sigma^1$ minimizes $C_1^j$,
  $\sum_{a \in M^j} c_1^j(\sigma^1,a)p(a) \leq \sum_{a \in M^j}
  C(\mathbb{A}_{OPT}, a)p(a)$.  This implies the statement for
  $\sigma^1$, and an analogous argument with $n-\beta_j+1$ zeros
  yields the statement for $\sigma^0$.
\end{proof}

Below we use Claims~\ref{clm:prob0prob1opt} and~\ref{clm:prob01ltopt}
in order to prove Theorem~\ref{thm:arbRR4approx}.

\begin{proof}[Proof of Thoerem~\ref{thm:arbRR4approx}]
  Let $\OPT$ be the expected cost
  incurred by an optimal strategy. We partition the set of all possible assignments
  $a\in\fullassign$ into two groups, $Q_{0}$ and $Q_{1}$,
  depending on whether running Algorithm~\ref{alg:arbRR} on $a$ causes it
to terminate after querying a bit chosen by
  $\mathrm{Alg}_{0}$ or a bit chosen by
  $\mathrm{Alg}_{1}$ (respectively).
  
For $l \in \bit$, let $C^{RR}(\mathrm{Alg}_{l},a)$
  represent the cost incurred by $\mathrm{Alg}_{l}$ during execution
  of Algorithm~\ref{alg:arbRR} on assignment $a$.
  As in Section~\ref{sec:MRR}, it holds that
  for $a\in Q_{0}$,
  $C^{RR}(\mathrm{Alg}_{0}, a)\geq C^{RR}(\mathrm{Alg}_{1}, a)$ and for
  $a\in Q_{1}$, $C^{RR}(\mathrm{Alg}_{1}, a)\geq C^{RR}(\mathrm{Alg}_{0},
  a)$.

  Suppose $a \in M^j \cap Q_1$.  Algorithm~\ref{alg:arbRR} terminates
  on input $a$ as soon as it has seen at least $\alpha_j$ ones and at
  least $n-\beta_{j} + 1$ zeros.  Since $a \in Q_1$,
  Algorithm~\ref{alg:arbRR} terminated as soon as it saw its $\alpha_{j}$th 1.
  It follows that $C^{RR}(\mathrm{Alg}_1,a) \leq c_1^j(\sigma^1,a)$.
Similarly, for $a \in M^j \cap Q_0 $,
$C^{RR}(\mathrm{Alg}_0,a) \leq c_0^j(\sigma^0,a)$.
Letting $B$ be the total number of blocks, so blocks are numbered from 1 to $B$,
Claim~\ref{clm:prob01ltopt} implies that for $l\in\bit$
\begin{equation}
    \label{eq:prob01ltopt}
  \begin{split}
    \sum_{j=1}^B \sum_{a\in M^j \cap Q_l}C^{RR}(\mathrm{Alg}_{l}, a)p(a)  
    &\leq
    \sum_{j=1}^B \sum_{a\in M^j \cap Q_l}c_l^j(\sigma^l, a)p(a)  \\
    &\leq
    \sum_{j=1}^B C_l^j(\sigma^l) 
    \leq
    \sum_{j=1}^B \sum_{a\in M^j}C(\mathbb{A}_{OPT},a)p(a) = \OPT
  \end{split}
\end{equation}
Thus, letting $EC$ be the expected cost of the Algorithm~\ref{alg:arbRR},
it follows from~(\ref{eq:prob01ltopt}) that we have
  \begin{equation}
    \label{eq:arb44expcost}
    \begin{split}
      EC &= \sum_{a\in Q_{0}} \left[C^{RR}(\mathrm{Alg}_{0},a) +
        C^{RR}(\mathrm{Alg}_{1},a)\right]p(a) \\
     &\qquad
     + \sum_{a\in Q_{1}} \left[C^{RR}(\mathrm{Alg}_{0},a) + C^{RR}(\mathrm{Alg}_{1},a)\right]p(a) \\
      &\leq 2\sum_{a\in Q_{0}} C^{RR}(\mathrm{Alg}_{0},a)p(a) 
      +2\sum_{a\in Q_{1}} C^{RR}(\mathrm{Alg}_{1},a)p(a) \\
      &\leq 2\sum_{j=1}^B\sum_{a\in M^j \cap Q_{0}} C^{RR}(\mathrm{Alg}_{0},a) p(a)
      +2\sum_{j=1}^B\sum_{a\in M^j\cap Q_{1}} C^{RR}(\mathrm{Alg}_{1},a) p(a) \\
      &\leq 2\OPT + 2\OPT = 4\OPT
  \end{split}
  \end{equation}
\end{proof}